\documentclass[11pt]{article}%
\newcommand{\remove}[1]{}

\usepackage{hyperref}  
\usepackage{graphicx}
\usepackage{amsmath,amssymb,latexsym}
\usepackage[in]{fullpage}%
\usepackage{amstext}
\usepackage{scalerel}
\usepackage{mleftright}%
\usepackage{xspace}%
\usepackage{euscript}%
\usepackage[cmyk]{xcolor}%

\hypersetup{%
      breaklinks,%
      ocgcolorlinks, colorlinks=true,%
      urlcolor=[rgb]{0.25,0.0,0.0},%
      linkcolor=[rgb]{0.5,0.0,0.0},%
      citecolor=[rgb]{0,0.2,0.445},%
      filecolor=[rgb]{0,0,0.4},
      anchorcolor=[rgb]={0.0,0.1,0.2}%
}

\newcommand{\nfrac}[2]{{#1}/{#2}}

\newcommand{\etal}{\textit{et~al.}\xspace}
\newcommand{\MdPrc}{\nu}
\newcommand{\MdOpt}{\nu_{\mathrm{opt}}}
\newcommand{\MdOptD}{\nu^D_{\mathrm{opt}}}
\newcommand{\RadMdOpt}{\mathrm{R^{\nu}_{opt}}}

\newcommand{\SqPrc}{\mu}
\newcommand{\SqOpt}{\mu_{\mathrm{opt}}}
\newcommand{\SqOptD}{\mu^D_{\mathrm{opt}}}
\newcommand{\RadSqOpt}{\mathrm{R^{\mu}_{opt}}}

\newcommand{\Dist}{\mathbf{d}\hspace{-1pt}}

\newcommand{\CSet}{A}
\newcommand{\CSetB}{B}
\newcommand{\CenSet}{\mathcal{D}}

\newcommand{\Copt}{C_{\mathrm{opt}}}

\newcommand{\constA}{\gamma}
\newcommand{\ball}{{b}}
\newcommand{\Pbad}{P_{\mathrm{bad}}}

\newcommand{\GridSqr}{Q}
\newcommand{\Coreset}{{\mathcal{S}}}
\newcommand{\BadBallSize}{\eta}
\providecommand{\polylog}{\mathrm{polylog}}%
\newcommand{\Ckr}{\varrho}
\newcommand{\CkrExp}{\exp{ [{ O \pth{ \nfrac{(1+\log
                           \nfrac{1}{\eps})}{\eps}}^{d-1}}] }}
\newcommand{\depth}{\mathrm{depth}}

\providecommand{\tildegen}{{\protect\raisebox{-0.1cm}
      {\symbol{'176}\hspace{-0.01cm}}}}
\newcommand{\atgen}{\symbol{'100}}

\newcommand{\eps}{{\varepsilon}}%
\newcommand{\pth}[2][\!]{\mleft({#2}\mright)}%
\providecommand{\SarielThanks}{}%
\renewcommand{\SarielThanks}{%
   \thanks{Department of Computer Science, DCL 2111; University of
      Illinois; 1304 West Springfield Ave.; Urbana, IL 61801; USA;
      {{\tt http://www.uiuc.edu/\tildegen{}sariel/}}; {{\tt
            sariel\atgen{}uiuc.edu}}.  Work on this paper was
      partially supported by a NSF CAREER award CCR-0132901.  }%
}%

\newcommand{\SohamThanks}{%
   \thanks{Department of Computer Science; University of
      Illinois; 1304 West Springfield Ave.; Urbana, IL 61801; USA;
      {{\tt smazumda\atgen{}uiuc.edu}}. }%
}%

\newcommand{\HLinkShort}[2]{\hyperref[#2]{#1\ref*{#2}}}
\newcommand{\HLink}[2]{\hyperref[#2]{#1~\ref*{#2}}}
\newcommand{\HLinkPage}[2]{\hyperref[#2]{#1~\ref*{#2}%
      $_\text{p\pageref{#2}}$}}
\newcommand{\HLinkPageOnly}[1]{\hyperref[#1]{Page~\refpage*{#1}%
      $_\text{p\pageref{#1}}$}}

\newcommand{\HLinkSuffix}[3]{\hyperref[#2]{#1\ref*{#2}{#3}}}
\newcommand{\HLinkPageSuffix}[3]{\hyperref[#2]{#1\ref*{#2}%
      #3$_\text{p\pageref{#2}}$}}
\providecommand{\si}[1]{#1}

\newcommand{\seclab}[1]{\label{sec:#1}}
\newcommand{\secref}[1]{\HLink{Section}{sec:#1}}

\newcommand{\corlab}[1]{\label{cor:#1}}
\newcommand{\corref}[1]{\HLink{Corollary}{cor:#1}}%

\newcommand{\apndlab}[1]{\label{apnd:#1}}
\newcommand{\apndref}[1]{\HLink{Appendix}{apnd:#1}}

\newcommand{\defrefY}[2]{\hyperref[def:#2]{#1}}

\newcommand{\lemlab}[1]{\label{lemma:#1}}
\newcommand{\lemref}[1]{\HLink{Lemma}{lemma:#1}}%

\newcommand{\tablab}[1]{\label{table:#1}}%
\newcommand{\tabref}[1]{\HLink{Table}{table:#1}}%

\newcommand{\thmlab}[1]{{\label{theo:#1}}}
\newcommand{\thmref}[1]{\HLink{Theorem}{theo:#1}}

\providecommand{\eqlab}[1]{}%
\renewcommand{\eqlab}[1]{\label{equation:#1}}

\newcommand{\tldO}{\scalerel*{\widetilde{O}}{j^2}}%
\providecommand{\Matousek}{Matou{\v s}ek\xspace}

\usepackage[amsmath,thmmarks]{ntheorem}%
\theoremseparator{.}%

\theoremstyle{plain}%
\newtheorem{theorem}{Theorem}[section]

\newtheorem{lemma}[theorem]{Lemma}

\newtheorem{corollary}[theorem]{Corollary}

\newtheorem{observation}[theorem]{Observation}

\theoremstyle{plain}%
\theoremheaderfont{\sf} \theorembodyfont{\upshape}%
\newtheorem*{remark:unnumbered}[FakeCounter]{Remark}%
\newtheorem{remark}[theorem]{Remark}%

\newtheorem*{defn:unnumbered}[FakeCounter]{Definition}
\newtheorem{defn}[theorem]{Definition}

\newcommand{\myqedsymbol}{\rule{2mm}{2mm}}

\theoremheaderfont{\em}%
\theorembodyfont{\upshape}%
\theoremstyle{nonumberplain}%
\theoremseparator{}%
\theoremsymbol{\myqedsymbol}%
\newtheorem{proof}{Proof:}%

\newcommand{\dist}[1]{\left\| {#1} \right\|}
\newcommand{\cardin}[1]{\left| {#1} \right|}%
\newcommand{\brc}[1]{\left\{ {#1} \right\}}
\newcommand{\ceil}[1]{\left\lceil {#1} \right\rceil}
\newcommand{\E}{\EuScript{E}}%

\newcommand{\Set}[2]{\left\{ #1 \;\middle\vert\; #2 \right\}}
\newcommand{\pbrc}[2][\!\!]{#1\left[ {#2} \bigr. \right]}

\newcommand{\hrefb}[3][black]{\href{#2}{\color{#1}{#3}}}%

\begin{document}

\title{Coresets for $k$-Means and $k$-Median Clustering and their
   Applications%
   \thanks{The latest version of the paper is available from the
      author webpage
      \href{http://valis.cs.uiuc.edu/~sariel/research/papers/03/kcoreset/}
      {http://valis.cs.uiuc.edu/\tildegen{}sariel/research/papers/03/kcoreset/}
      .} }

\author{Sariel Har-Peled%
   \SarielThanks{}%
   \and %
   Soham Mazumdar%
   \SohamThanks{}%
}

\date{November 7, 2003\footnote{Re\LaTeX{}ed and refs updated on October
      30, 2018.}}

\maketitle

\begin{abstract}
    In this paper, we show the existence of small coresets for the
    problems of computing $k$-median and $k$-means clustering for
    points in low dimension. In other words, we show that given a
    point set $P$ in $\Re^d$, one can compute a weighted set
    $\Coreset \subseteq P$, of size $O(k \eps^{-d} \log{n})$, such
    that one can compute the $k$-median/means clustering on $\Coreset$
    instead of on $P$, and get an $(1+\eps)$-approximation.
    
    As a result, we improve the fastest known algorithms for
    $(1+\eps)$-approximate $k$-means and $k$-median. Our
    algorithms have \emph{linear} running time for a fixed
    $k$ and $\eps$.  In addition, we can maintain the
    $(1+\eps)$-approximate $k$-median or $k$-means
    clustering of a stream when points are being \emph{only
       inserted}, using polylogarithmic space and update
    time.
\end{abstract}


\section{Introduction}

Clustering is a widely used technique in Computer Science with
applications to unsupervised learning, classification, data mining
and other fields. We study two variants of the clustering problem
in the geometric setting. The \emph{geometric $k$-median
clustering} problem is the following: Given a set $P$ of points in
$\Re^d$, compute a set of $k$ points in $\Re^d$ such that the sum
of the distances of the points in $P$ to their respective nearest
median is minimized. The $k$-means differs from the above in that
instead of the sum of distances, we minimize the sum of squares of
distances.  Interestingly the $1$-mean is the center of mass of
the points, while the $1$-median problem, also known as the
Fermat-Weber problem, has no such closed form. As such the problems
have usually been studied separately from each other even in the
approximate setting. We propose techniques which can be used for
finding approximate $k$ centers in both variants.

In the data stream model of computation, the points are read
in a sequence and we desire to compute a function,
clustering in our case, on the set of points seen so far. In
typical applications, the total volume of data is very large
and can not be stored in its entirety. Thus we usually
require a data-structure to maintain an aggregate of the
points seen so far so as to facilitate computation of the
objective function. Thus the standard complexity measures in
the data stream model are the storage cost, the update cost
on seeing a new point and the time to compute the function
from the aggregated data structure.

\begin{table}[t]
    \centerline{
    }
    \begin{tabular}{|l||l|l|r|r|r|r|}
      \hline
      Problem
      & Previous Results
      & Our Results
      \\
      \hline
      \hline
      $k$-median
      &        
        \begin{minipage}{2.5in}
            \smallskip%
            $O(\Ckr n(\log{n})\log {k})$ \cite{kr-nltas-99} ~~{(*)}\\
               $~~~~~$ $\Ckr = \CkrExp$\\
           \end{minipage}
           &
           \begin{minipage}{2.5in}
           $O \pth{ n +
              \Ckr k^{O(1)} \log^{O(1)} n }$\\
           $[$\thmref{fast:k:median}$]$
           \end{minipage}\\      
      \hline
           \begin{minipage}{0.8in}
           discrete\\ $k$-median
           \end{minipage}&
               $O(\Ckr n\log{n}\log {k})$ \cite{kr-nltas-99}
           &
             \begin{minipage}{1.8in}
                 \smallskip%
                 $O \pth{ n +
                    \Ckr k^{O(1)} \log^{O(1)} n }$\\
                 \smallskip%
                 $[$\thmref{approx:k:median:discrete}$]$
             \end{minipage}\\[0.1cm]
           \hline
           \begin{minipage}{0.8in}
               $k$-means
           \end{minipage}&
           $O_k(n (\log{n})^k \eps^{-2k^2 d})$
           \cite{m-agc-00} ~~{(**)}
           &
           \begin{minipage}{2.4in}
               \smallskip%
               $O (n +{k^{k+2}}{\eps^{-(2d+1)k}} {\log^{k+1}{n}}\log^k\frac{1}{\eps})$\\
               $[$\thmref{k:means:approx}$]$
               \smallskip%
           \end{minipage}\\
      \hline
           \begin{minipage}{0.8in}
               Streaming
           \end{minipage}
           &
           \begin{minipage}{2.4in}
               $k$-median\\
               \si{Const} factor; Any metric space\\
               $O( k \polylog )$ space ~\cite{cop-bsacp-03}
           \end{minipage}&
           \begin{minipage}{2in}
               \smallskip%
               $k$-means and $k$-median\\
               $(1+\eps)$-approx; Points in $\Re^d$. \\
               $O( k \eps^{-d} \log^{2d+2}
               n )$ space\\
               $[$\thmref{k:means:streaming}$]$

               \smallskip
           \end{minipage}\\
           \hline
       \end{tabular}%
    \caption{For our results, all the running time bounds are in
       expectation and the algorithms succeed with high
       probability. (*) Getting this running time
       requires non-trivial modifications of the algorithm
       of Kolliopoulos and Rao \cite{kr-nltas-99}. (**) The
       $O_k$ notation hides constants that depends solely on
       $k$. }
    \vspace{-0.4cm}
    \tablab{results}
\end{table}

\smallskip
\noindent{\bf $k$-median clustering.} The $k$-median problem turned
out to be nontrivial even in low dimensions and achieving a good
approximation proved to be a challenge. Motivated by the work of Arora
\cite{a-ptase-98}, which proposed a new technique for geometric
approximation algorithms, Arora, Raghavan and Rao \cite{arr-asekm-98}
presented a $O\pth{n^{O(1/\eps)+1}}$ time $(1+\eps)$-approximation
algorithm for points in the plane. This was significantly improved by
Kolliopoulos and Rao \cite{kr-nltas-99} who proposed an algorithm with
a running time of $O(\Ckr n\log{n}\log {k})$ for the discrete version
of the problem, where the medians must belong to the input set and
$\Ckr = \CkrExp$. The $k$-median problem has been studied extensively
for arbitrary metric spaces and is closely related to the
uncapacitated facility location problem.  Charikar \etal{}
\cite{cgts-cfaak-99} proposed the first constant factor approximation
to the problem for an arbitrary metric space using a natural linear
programming relaxation of the problem followed by rounding the
fractional solution. The fastest known algorithm is due to Mettu and
Plaxton \cite{mp-otbap-02} who give an algorithm which runs in
$O(n(k+\log n))$ time for small enough $k$ given the distances are
bound by $2^{O\pth{n/ \log\pth{n/k}}}$. It was observed that if the
constraint of having exactly $k$-medians is relaxed, the problem
becomes considerably easier \cite{cg-icaflkm-99,jv-pdaam-99}. In
particular, Indyk \cite{i-stams-99} proposed a constant factor
approximation algorithm which produces $O(k)$ medians in
$\tldO \pth{n k}$ time. In the streaming context, Guha \etal {}
\cite{gmmo-cds-00} propose an algorithm which uses $O(n^\eps)$ memory
to compute $2^{1/\eps}$ approximate $k$-medians.  Charikar \etal {}
\cite {cop-bsacp-03} improve the algorithm by reducing the space
requirement to $O(k \cdot \polylog(n))$.

\smallskip
\noindent{\bf $k$-means clustering.} Inaba \etal{} \cite{iki-awvdr-94}
observe that the number of Voronoi partitions of $k$ points in $\Re^d$
is $n^{k d}$ and can be done exactly in time $O(n^{k d+1})$. They also
propose approximation algorithms for the $2$-means clustering problem
with time complexity $O(n^{O(d)})$.  \si{de} la Vega \etal{}
\cite{vkkr-ascp-03} proposes a $(1+\eps)$-approximation algorithm, for
high dimensions, with running time $O(g(k,\eps) n \log^k n)$, where
$g(k,\eps) = \exp [ (k^3/\eps^8) (\ln(k/\eps)) \ln k]$ (they refer to
it as $\ell_2^2$ $k$-median clustering).  \Matousek{} \cite{m-agc-00}
proposed a $(1+\eps)$-approximation algorithm for the geometric
$k$-means problem with running time $O\pth{n\eps^{-2k^2d}\log^k{n} }$.

\paragraph{Our Results.}
We propose fast algorithms for the approximate $k$-means and
$k$-medians problems.  The central idea behind our
algorithms is computing a weighted point set which we call a
\emph{$\pth{k,\eps}$-coreset}. For an optimization problem,
a coreset is a subset of input, such that we can get a good
approximation to the original input by solving the
optimization problem directly on the coreset. As such, to
get good approximation, one needs to compute a coreset, as
small as possible from the input, and then solve the problem
on the coreset using known techniques. Coresets have been
used for geometric approximation mainly in low-dimension
\cite{ahv-aemp-04,h-cm-04,apv-aaklc-02}, although a similar
but weaker concept was also used in high dimensions
\cite{bhi-accs-02,bc-ocsb-03,hv-pchdu-02}. In low dimensions
coresets yield approximation algorithm with linear or near
linear running time with an additional term that depends
only on the size of the coreset.

In the present case, the property we desire of the
$\pth{k,\eps}$-coreset is that the clustering cost of the coreset for
any arbitrary set of $k$ centers is within $\pth{1\pm \eps}$ of the
cost of the clustering for the original input. To facilitate the
computation of the coreset, we first show a linear time algorithm (for
$k =O(n^{1/4})$), that constructs a $O( k\polylog )$ sized set of
centers such that the induced clustering gives a constant factor
approximation to both the optimal $k$-means and the optimal
$k$-medians clustering. We believe that the technique used for this
fast algorithm is of independent interest. Note that it is faster than
previous published fast algorithms for this problem (see
\cite{mp-otbap-02} and references therein), since we are willing to
use more centers.  Next, we show how to construct a suitably small
coreset from the set of approximate centers.  We compute the $k$
clusterings for the coresets using weighted variants of known
clustering algorithms. Our results are summarized in \tabref{results}.

One of the benefits of our new algorithms is that in the resulting
bounds, on the running time, the term containing `$n$' is decoupled
from the ``nasty'' exponential constants that depend on $k$ and
$1/\eps$. Those exponential constants seems to be inherent to the
clustering techniques currently known for those problems.

Our techniques extend very naturally to the streaming model
of computation. The aggregate data-structure is just a
$(k,\eps)$-coreset of the stream seen so far. The size of
the maintained coreset is $O( k\eps^{-d} \log{n})$, and the
overall space used is $O((\log^{2d+2} {n})/\eps^d )$. The
amortized time to update the data-structure on seeing a new
point is $O(k^5+\log^2(k/\eps))$.

As a side note, our ability to get linear time algorithms for fixed
$k$ and $\eps$, relies on the fact that our algorithms need to solve a
batched version of the nearest neighbor problem. In our algorithms,
the number of queries is considerably larger than the number of sites,
and the distances of interest arise from clustering. Thus, a small
additive error which is related to the total price of the clustering
is acceptable. In particular, one can build a data-structure that
answers nearest neighbor queries in $O(1)$ time per query, see
\apndref{fast:n:n}. Although this is a very restricted case, this
result may nevertheless be of independent interest, as this is the
first data-structure to offer nearest neighbor queries in constant
time, in a non-trivial settings.

The paper is organized as follows. In \secref{coresets}, we
prove the existence of coresets for $k$-median/means
clustering. In \secref{fast:const:factor}, we describe the
fast constant factor approximation algorithm which generates
more than $k$ means/medians. In \secref{eps:approx:k:median}
and \secref{eps:approx:k:mean}, we combine the results of the two preceding
sections, and present an $(1+\eps)$-approximation algorithm
for $k$-means and $k$-median respectively. In
\secref{streaming}, we show how to use coresets for space
efficient streaming. We conclude in \secref{conclusions}.

\section{Preliminaries}

\begin{defn}
    For a point set $X$, and a point $p$, both in $\Re^d$, let
    $\Dist(p,X)= \min_{x \in X} \dist{x p}$ denote the \emph{distance
       of $p$ from $X$}.
\end{defn}

\begin{defn}[Clustering]
    For a weighted point set $P$ with points from $\Re^d$, with an
    associated weight function $w:P\rightarrow \mathbb{Z}^+$ and any
    point set $C$, we define
    $\MdPrc_C\pth{P} = \sum_{p\in P} w(p)\Dist(p,C)$ as the
    \emph{price} of the $k$-median clustering provided by $C$. Further
    let
    $\MdOpt(P,k) =\min_{C\subseteq \Re^d, \cardin{C} = k}
    \MdPrc_C\pth{P}$ denote the price of the \emph{optimal $k$-median}
    clustering for $P$.

    Similarly, let $\SqPrc_C(P) = \sum_{p \in P}
    w(p)\pth{\Dist(p,C)}^2$ denote the price of the
    $k$-means clustering of $P$ as provided by the set of
    centers $C$.  Let $\SqOpt(P,k) = \min_{C \subseteq
       \Re^d, \cardin{C} = k} \SqPrc_C(P)$ denote the price
    of the \emph{optimal $k$-means clustering} of $P$.
\end{defn}

\begin{remark}
    We only consider positive integer weights. A regular
    point set $P$ may be considered as a weighted set with weight $1$
    for each point, and total weight $\cardin{P}$.
\end{remark}

\begin{defn}[Discrete Clustering]
    In several cases, it is convenient to consider the
    centers to be restricted to lie in the original point
    set. In particular, let $\MdOptD(P,k) =\min_{C\subseteq
       P, \cardin{C} = k} \MdPrc_C\pth{P}$ denote the price
    of the \emph{optimal discrete $k$-median} clustering for
    $P$ and let $\SqOptD(P,k) = \min_{C \subseteq P,
       \cardin{C} = k} \SqPrc_C(P)$ denote the price of the
    \emph{optimal discrete $k$-means} clustering of $P$.
\end{defn}

\begin{observation}
    For any point set $P$, we have $\SqOpt(P,k) \leq \SqOptD(P,k) \leq
    4 \SqOpt(P,k)$, and $\MdOpt(P,k) \leq \MdOptD(P,k) \leq 2
    \MdOpt(P,k)$.
\end{observation}


\section{Coresets from Approximate Clustering}
\seclab{coresets}

\begin{defn}[Coreset]
    For a weighted point set $P\subseteq \Re^d$, a weighted
    set $\Coreset \subseteq \Re^d$, is a
    \emph{$\pth{k,\eps}$-coreset} of $P$ for the $k$-median
    clustering, if for any set $C$ of $k$ points in $\Re^d$,
    we have $(1-\eps)\MdPrc_C(P) \leq \MdPrc_C(\Coreset)
    \leq (1+\eps)\MdPrc_C(P)$.

    Similarly, $\Coreset$ is a \emph{$\pth{k,\eps}$-coreset}
    of $P$ for the $k$-means clustering, if for any set $C$
    of $k$ points in $\Re^d$, we have $(1-\eps)\SqPrc_C(P)
    \leq \SqPrc_C(\Coreset) \leq (1+\eps)\SqPrc_C(P)$.
\end{defn}

\subsection{Coreset for $k$-Median}

Let $P$ be a set of $n$ points in $\Re^d$, and
$\CSet = \brc{x_1,\ldots, x_m}$ be a point set, such that
$\MdPrc_\CSet(P) \leq c \MdOpt(P,k)$, where $c$ is a constant. We give
a construction for a $(k,\eps)$-coreset using $\CSet$. Note that we do
not have any restriction on the size of $\CSet$, which in subsequent
uses will be taken to be $O(k\polylog)$.

\subsubsection{The construction}
\seclab{construction}

Let $P_i$ be the points of $P$ having $x_i$ as their nearest
neighbor in $\CSet$, for $i=1,\ldots, m$. Let $R =
\MdPrc_\CSet(P)/\pth{c n}$ be a lower bound estimate of the
average radius $\RadMdOpt(P,k) = \MdOpt(P,k)/n$.  For any $p
\in P_i$, we have $\dist{p x_i} \leq c n R$, since $\dist{p
   x_i} \leq \MdPrc_\CSet(P)$, for $i=1,\ldots, m$.

Next, we construct an appropriate exponential grid around each $x_i$,
and snap the points of $P$ to those grids. Let $\GridSqr_{i,j}$ be an
axis-parallel square with side length $R 2^j$ centered at $x_i$, for
$j=0,\ldots, M$, where $M=\ceil{2\lg (c n)}$.  Next, let
$V_{i,0} = \GridSqr_{i,0}$, and let
$V_{i,j} = \GridSqr_{i,j} \setminus \GridSqr_{i,j-1}$, for
$j=1,\ldots, M$.  Partition $V_{i,j}$ into a grid with side length
$r_{j} = \eps R 2^j/(10 c d)$, and let $G_i$ denote the resulting
exponential grid for $V_{i,0},\ldots, V_{i,M}$.  Next, compute for
every point of $P_i$, the grid cell in $G_i$ that contains it.  For
every non empty grid cell, pick an arbitrary point of $P_i$ inside it
as a representative point for the coreset, and assign it a weight
equal to the number of points of $P_i$ in this grid cell. Let
$\Coreset_i$ denote the resulting weighted set, for $i=1,\ldots, m$,
and let $\Coreset = \cup_i \Coreset_i$.

Note that $\cardin{\Coreset} = O\pth{\pth{\cardin{\CSet}\log n}/
   \eps^d }$.  As for computing $\Coreset$ efficiently. Observe that
all we need is a constant factor approximation to $\MdPrc_\CSet(P)$
(i.e., we can assign a $p \in P$ to $P_i$ if $\dist{p, x_i} \leq
2\Dist(p,\CSet)$).  This can be done in a naive way in $O(n m)$ time,
which might be quite sufficient in practice. Alternatively, one can
use a data-structure that answers constant approximate
nearest-neighbor queries in $O(\log m)$ when used on $\CSet$ after
$O(m \log{m})$ preprocessing \cite{amnsw-oaann-98}.  Another option
for computing those distances between the points of $P$ and the set
$\CSet$ is using \thmref{batch:n:n} that works in $O(n + m n^{1/4}
\log{n})$ time.  Thus, for $i=1,\ldots, m$, we compute a set $P_i'$
which consists of the points of $P$ that $x_i$ (approximately) serves.
Next, we compute the exponential grids, and compute for each point of
$P_i'$ its grid cell.  This takes $O(1)$ time per point, with a
careful implementation, using hashing, the floor function and the
$\log$ function. Thus, if $m=O(\sqrt{n})$ the overall running time is
$O(n + m n^{1/4} \log{n}) = O(n)$ and $O( m\log{m} + n \log{m} + n) =
O( n\log{m})$ otherwise.


\subsubsection{Proof of Correctness}

\begin{lemma}
    \lemlab{correctness}%
    The weighted set $\Coreset$ is a $\pth{k,\eps}$-coreset for $P$
    and $\cardin{\Coreset} = O\pth{|\CSet|\eps^{-d}\log{n}}$.
\end{lemma}

\begin{proof}
    Let $Y$ be an arbitrary set of $k$ points in $\Re^d$.
    For any $p\in P$, let $p'$ denote the image of $p$ in
    $\Coreset$. The error is $\E =
    \cardin{\MdPrc_Y\pth{P} - \MdPrc_Y(\Coreset)} \leq
    \sum_{p\in P}\cardin{\Dist(p, Y) - \Dist(p', Y)}$.

    Observe that $\Dist(p,Y) \leq \dist{p p'} +\Dist(p',Y)$
    and $\Dist(p',Y) \leq \dist{p p'} +\Dist(p,Y)$ by the
    triangle inequality. Implying that $\cardin{\Dist(p, Y)
       - \Dist(p',Y)} \leq \dist{p p'}$.  It follows that
    \[
    \E \leq \sum_{p\in P}\dist{p p'} =
    {\!\!\!\!\!\sum_{\substack{p\in P,\\
             \Dist(p,\CSet) \leq R} } \!\!\!\!\! \dist{p p'}
    } +
    \!\!\!\!\!{\sum_{\substack{p\in P,\\
             \Dist(p,\CSet) > R} } \!\!\!\!\!  \dist{p p'} }
    \leq \frac{\eps}{10c} n R +\frac{\eps}{10c}\sum_{p\in P}
    \Dist\pth{p,\CSet} \leq \frac{2\eps}{10c}
    \MdPrc_\CSet(P) \leq \eps \MdOpt\pth{P,k},
    \]
    since $\dist{p p'} \leq \frac{\eps}{10c}\Dist(p,A)$ if
    $\Dist(p,A) \geq R$, and $\dist{p p'} \leq
    \frac{\eps}{10c} R$, if $\Dist(p,A) \leq R$, by the
    construction of the grid. This implies
    $\cardin{\MdPrc_Y\pth{P} - \MdPrc_Y(\Coreset)} \leq \eps
    \MdPrc_Y \pth{ P}$, since $\MdOpt(P,k) \leq
    \MdPrc_Y(P)$.
\end{proof}

It is easy to see that the above algorithm can be easily
extended for weighted point sets.
\begin{theorem}
    \thmlab{coreset:fast:k:median}%
    Given a point set $P$ with $n$ points, and a point set $\CSet$
    with $m$ points, such that $\MdPrc_\CSet(P) \leq c \MdOpt(P,k)$,
    where $c$ is a constant. Then, one can compute a weighted set
    $\Coreset$ which is a $(k,\eps)$-coreset for $P$, and
    $\cardin{\Coreset} = O\pth{ (\cardin{\CSet} \log{n}) /\eps^d}$.
    The running time is $O(n )$ if $m=O(\sqrt{n})$ and $O( n\log{m})$
    otherwise.

    If $P$ is weighted, with total weight $W$, then
    $\cardin{\Coreset} = O\pth{ (\cardin{\CSet} \log{W})
       /\eps^d}$.

\end{theorem}

\subsection{Coreset for $k$-Means}
\seclab{k:means:coreset}

The construction of the $k$-means coreset follows the $k$-median
with a few minor modifications.  Let $P$ be a set of $n$ points in
$\Re^d$, and a $\CSet$ be a point set $\CSet = \brc{x_1,\ldots,
x_m}$, such that $\SqPrc_\CSet(P) \leq c \SqOpt(P,k)$.  Let $R =
\sqrt{(\SqPrc_\CSet(P)/(c n))}$ be a lower bound estimate of the
average mean radius $\RadSqOpt(P,k) = \sqrt{ \SqOpt(P,k)/n}$.  For
any $p \in P_i$, we have $\dist{p x_i} \leq \sqrt{c n} R$, since
$\dist{p
   x_i}^2 \leq \SqPrc_\CSet(P)$, for $i=1,\ldots, m$.

Next, we construct an exponential grid around each point of $\CSet$,
as in the $k$-median case, and snap the points of $P$ to this grid,
and we pick a representative point for such grid cell. See
\secref{construction} for details.  We claim that the resulting set of
representatives $\Coreset$ is the required coreset.

\begin{theorem}
    \thmlab{k:means:weighted:coreset}%
    Given a set $P$ with $n$ points, and a point set $\CSet$ with $m$
    points, such that $\SqPrc_\CSet(P) \leq c \SqOpt(P,k)$, where $c$
    is a constant. Then, can compute a weighted set $\Coreset$ which
    is a $(k,\eps)$-coreset for $P$, and $\cardin{\Coreset} = O\pth{
       (m \log{n}) /(c\eps)^d}$.  The running time is $O(n )$ if
    $m=O(n^{1/4})$ and $O( n\log{m})$ otherwise.
    
    If $P$ is a weighted set with total weight $W$, then the
    size of the coreset is $O\pth{ (m \log{W})/\eps^d}$.
\end{theorem}

\begin{proof}
    We prove the theorem for an unweighted point set.  The
    construction is as in \secref{k:means:coreset}.  As for
    correctness, consider an arbitrary set $\CSetB$ of $k$ points in
    $\Re^d$.  The proof is somewhat more tedious than the median case,
    and we give short description of it before plunging into the
    details. We partition the points of $P$ into three sets: (i)
    Points that are close (i.e., $\leq R$) to both $\CSetB$ and
    $\CSet$. The error those points contribute is small because they
    contribute small terms to the summation.  (ii) Points that are
    closer to $\CSetB$ than to $\CSet$ (i.e., $P_\CSet$). The error
    those points contribute can be charged to an $\eps$ fraction of
    the summation $\SqPrc_\CSet(P)$. (iii) Points that are closer to
    $\CSet$ than to $\CSetB$ (i.e., $P_\CSetB$). The error is here
    charged to the summation $\SqPrc_\CSetB(P)$.  Combining those
    three error bounds, give us the required result.

    For any $p\in P$, let $p'$ the image of $p$ in $\Coreset$; namely,
    $p'$ is the point in the coreset $\Coreset$ that represents $p$.
    Now, we have
    \begin{equation*}
        \E%
        =%
        \cardin{\SqPrc_\CSetB(P) - \SqPrc_\CSetB(S) } \leq
        \sum_{p\in P} \cardin{ {\Dist(p,\CSetB) }^2
           - \Dist(p', \CSetB)^2}%
        \leq
        \sum_{p\in P} \cardin{ \pth{\Dist(p, \CSetB)
              - \Dist(p', \CSetB)}
           \pth{\Bigl.\!\Dist(p, \CSetB) + \Dist(p',B)} }
    \end{equation*}
    Let
    $P_R = \Set{ p \in P }{ \Dist(p, \CSetB ) \leq R \text{ and }
       \Dist(p , \CSet)\leq R}$,
    $P_\CSet = \Set{ p \in P \setminus P_R }{ \Dist(p,\CSetB) \leq
       \Dist(p,\CSet)}$, and let
    $P_\CSetB = P \setminus \pth{ P_R \cup P_\CSet}$. By the triangle
    inequality, for $p \in P$, we have
    $\Dist(p',\CSetB) + \dist{p p'} \geq \Dist(p, \CSetB)$ and
    $\Dist(p, \CSetB) + \dist{p p'} \geq \Dist(p',\CSetB)$.  Thus,
    $\dist{p p'} \geq \cardin{ \Dist(p, \CSetB) - \Dist(p', \CSetB)}$.
    
    Also,
    $\Dist(p, \CSetB) + \Dist(p', \CSetB) \leq 2\Dist(p, \CSetB) +
    \dist{p p'}$, and thus
    \begin{align*}
      \E_R%
      &=%
        \sum_{p\in P_R} \cardin{ \pth{ \Dist(p, \CSetB) -
        \Dist(p',\CSetB)} \pth{\Dist(p, \CSetB) + \Dist(p',B)}
        }%
      \leq %
      \sum_{p\in P_R} \dist{p p'} \pth{ 2 \Dist(p, \CSetB) + \dist{p
      {}p'}}%
      \\&%
      \leq%
      \sum_{p\in P_R} \frac{\eps}{10} R \pth{ 2 R
      + \frac{\eps}{10} R}%
      \leq %
      \frac{\eps}{3} \sum_{p\in P_R} R^2 \leq \frac{\eps}{3}
      \SqOpt(P,k),
    \end{align*}
    since by definition, for $p \in P_R$, we have
    $\Dist(p, \CSet), \Dist(p,\CSetB) \leq R$.
    
    By construction $\dist{p p'} \leq (\eps/10c)\Dist(p,\CSet)$, for
    all $p \in P_A$, as $\Dist(p, \CSet) \geq R$. Thus,
    \begin{align*}
      \E_\CSet
      &=%
        \sum_{p\in P_\CSet} \dist{p p'} \pth{ 2
        \Dist(p, \CSetB) + \dist{p p'}}
        \leq \sum_{p\in P_\CSet} \frac{\eps}{10c}
        \Dist(p,\CSet) \pth{ 2 + \frac{\eps}{10c}}
        \Dist(p, \CSet) %
      \\&%
      \leq%
      \frac{\eps}{3c} \sum_{p \in
      P_\CSet} \pth{ \Dist(p,\CSet)}^2 \leq
      \frac{\eps}{3} \SqOpt(P,k)
      \leq \frac{\eps}{3} \SqPrc_\CSetB(P).
    \end{align*}

    As for $p \in P_\CSetB$, we have $\dist{p p'} \leq
    \frac{\eps}{10c}\Dist(p, \CSetB)$, since $\Dist(p, \CSetB) \geq
    R$, and $\Dist(p,\CSet) \leq \Dist(p, \CSetB)$. Implying
    $\dist{p p'} \leq (\eps/10c)\Dist(p,B)$ and thus
    \begin{align*}
        \E_\CSetB
      &=%
        \sum_{p\in P_\CSetB} \dist{p p'} \pth{ 2
           \Dist(p, \CSetB) + \dist{p p'}}
        \leq
        \sum_{p \in P_\CSetB}
        \frac{\eps}{10c}\Dist(p,\CSetB)
        \pth{ 2
        \Dist(p, \CSetB) + \frac{\eps}{10c}\Dist(p,\CSetB)}%
      \\&%
      \leq%
        \sum_{p \in P_\CSetB} \frac{\eps}{3} \Dist(p, \CSetB)^2
        \leq \frac{\eps}{3} \SqPrc_\CSetB(P).
    \end{align*}
    We conclude that $\E = \cardin{\SqPrc_\CSetB(P) - \SqPrc_\CSetB(S)
    } \leq \E_R + \E_\CSet + \E_\CSetB \leq \frac{3\eps}{3}
    \SqPrc_\CSetB(P), $ which implies that $(1-\eps)\SqPrc_\CSetB(P)
    \leq \SqPrc_\CSetB(S) \leq (1+\eps) \SqPrc_\CSetB(P)$, as
    required. It is easy to see that we can extend the analysis
    for the case when we have weighted points.
\end{proof}


\section{Fast Constant Factor Approximation Algorithm}
\seclab{fast:const:factor}

Let $P$ be the given point set in $\Re^d$. We want to quickly
compute a constant factor approximation to the $k$-means
clustering of $P$, while using more than $k$ centers. The number
of centers output by our algorithm is $O\pth{k \log^3 n}$.
Surprisingly, the set of centers computed by the following
algorithm is a good approximation for both $k$-median and
$k$-means. To be consistent, throughout this section, we refer to
$k$-means, although everything holds nearly verbatim for
$k$-median as well.

\begin{defn}[bad points]
    For a point set $X$, define a point $p \in P$ as
    \emph{bad} with respect to $X$, if the cost it pays in
    using a center from $X$ is prohibitively larger than the
    price $\Copt$ pays for it; more precisely $\Dist(p,X)
    \geq 2 \Dist(p,\Copt)$. A point $p \in P$ which is not
    bad, is by necessity, if not by choice, \emph{good}.
    Here $\Copt =\Copt(P,k)$ is a set of optimal $k$-means
    centers realizing $\SqOpt(P,k)$.
\end{defn}
We first describe a procedure which given $P$, computes a
small set of centers $X$ and a large $P'\subseteq P$ such
that $X$ induces clusters $P'$ well. Intuitively we want a
set $X$ and a large set of points $P'$ which are \emph{good}
for $X$.

\subsection{Construction of the Set $X$ of Centers}

\seclab{good:subset:centers}

For $k=O(n^{1/4})$, we can compute a $2$-approximate $k$-center
clustering of $P$ in linear time \cite{h-cm-04}, or alternatively, for
$k=\Omega(n^{1/4})$, in $O(n\log{k})$ time, using the algorithm of
Feder and Greene \cite{fg-oafac-88}.  This is the \emph{min-max
   clustering} where we cover $P$ by a set of $k$ balls such the
radius of the largest ball is minimized. Let $V$ be the set of $k$
centers computed, together with the furthest point in $P$ from those
$k$ centers.

Let $L$ be the radius of this $2$-approximate clustering.  Since both
those algorithms are simulating the (slower) algorithm of Gonzalez
\cite{g-cmmid-85}, we have the property that the minimal distance
between any points of $V$ is at least $L$.  Thus, any $k$-means
clustering of $P$, must have price at least $(L/2)^2$, and is at most
of price $n L^2$, and as such $L$ is a rough estimate of
$\SqOpt(P,k)$. In fact, this holds even if we restrict out attention
only to $V$; explicitly $(L/2)^2 \leq \SqOpt(V, k) \leq \SqOpt(P,k)
\leq n L^2$.

Next, we pick a random sample $Y$ from $P$ of size $\rho = \constA k
\log^2 n$, where $\constA$ is a large enough constant whose value
would follow from our analysis.  Let $X = Y \cup V$ be the required
set of cluster centers. In the extreme case where $\rho > n$, we just
set $X$ to be $P$.

\subsection{A Large Good Subset for $X$}
\seclab{good:subset:points}

\subsubsection{Bad points are few}
Consider the set $\Copt$ of $k$ optimal centers for the $k$-means, and
place a ball $\ball_i$ around each point of $c_i\in\Copt$, such that
$\ball_i$ contain $\BadBallSize = n/(20k\log{n})$ points of $P$. If
$\constA$ is large enough, it is easy to see that with high
probability, there is at least one point of $X$ inside every ball
$\ball_i$. Namely, $X \cap \ball_i \ne \emptyset$, for $i=1,\ldots,
k$.

Let $\Pbad$ be the set of all bad points of $P$. Assume, that
there is a point $x_i \in X$ inside $\ball_i$, for $i=1,\ldots,
k$.  Observe, that for any $p \in P \setminus \ball_i$, we have
$\dist{p x_i} \leq 2\dist{p c_i}$. In particular, if $c_i$ is the
closest center in $\Copt$ to $p$, we have that $p$ is good.  Thus,
with high probability, the only bad points in $P$ are the one that
lie inside the balls $\ball_1,\ldots, \ball_k$. But every one of
those balls, contain at most $\BadBallSize$ points of $P$. It
follows, that with high probability, the number of bad points in
$P$ with respect to $X$ is at most $\beta= k \cdot \eta =
n/(20\log n)$.

\subsubsection{Keeping Away from Bad Points}

Although the number of bad points is small, there is no easy way to
determine the set of bad points. We instead construct a set $P'$
ensuring that the clustering cost of the bad points in $P'$ does not
dominate the total cost.  For every point in $P$, we compute its
approximate nearest neighbor in $X$.  This can be easily done in
$O(n\log \cardin{X} + \cardin{X} \log \cardin{X})$ time using
appropriate data structures \cite{amnsw-oaann-98}, or in $O(n +
n\cardin{X}^{1/4} \log n)$ time using \corref{batch:n:n:2} (with $D=n
L$). This stage takes $O(n)$ time, if $k=O(n^{1/4})$, else it takes
$O(n \log{\cardin{X}} + \cardin{X} \log{\cardin{X}}) = O(n \log( k\log
n))$ time, as $\cardin{X} \leq n$.

In the following, to simplify the exposition, we assume that we
compute exactly the distance $r(p) = \Dist(p,X)$, for $p \in P$.

Next, we partition $P$ into classes in the following way.  Let
$P[a,b] = \Set{ p \in P }{a\leq r(p)< b}$. Let $P_0 = P[0, L/(4n)]$,
$P_\infty = P[2Ln,\infty]$ and $P_i = P\pbrc{2^{i-1} L/n, 2^{i} L/n}$,
for $i=1,\ldots, M$, where $M = 2\ceil{\lg n} + 3$. This partition of
$P$ can be done in linear time using the $\log$ and floor function.

Let $P_\alpha$ be the last class in this sequence that
contains more than $2\beta = 2(n/(20\log{n}))$ points. Let
$P' = V \cup \bigcup_{i \leq \alpha} P_i$. We claim that
$P'$ is the required set.  Namely, $\cardin{P'} \geq n/2$
and $\SqPrc_X(P') = O(\SqPrc_{\Copt}(P'))$, where $\Copt
=\Copt(P,k)$ is the optimal set of centers for $P$.

\subsubsection{Proof of Correctness}

Clearly the set $P'$ contains at least $\pth{n - \cardin{P_\infty}
- M \cdot \pth{2n/20 \log{n}}}$ points. Since $P_\infty \subseteq
\Pbad$ and $\cardin{\Pbad} \leq \beta$, hence $|P'| > n/2$.

If $\alpha > 0$, we have $\cardin{P_\alpha} \geq 2\beta =
2(n/(20\log{n}))$. Since $P'$ is the union of all the classes with
distances smaller than the distances in $P_\alpha$, it follows that
the worst case scenario is when all the bad points are in $P_\alpha$.
But with high probability the number of bad points is at most $\beta$,
and since the price of all the points in $P_\alpha$ is roughly the
same, it follows that we can charge the price of the bad points in
$P'$ to the good points in $P_\alpha$.

Formally, let $Q'= P_\alpha \setminus \Pbad$.  For any point $p
\in P' \cap \Pbad$ and $q \in Q'$, we have $\Dist(p,X) \leq
2\Dist(q,X)$. Further $|Q'| > |\Pbad|$. Thus, $\SqPrc_X(P' \cap
\Pbad) \leq 4\SqPrc_X(Q') \leq 16 \SqPrc_{\Copt}(Q') \leq
16\SqPrc_{\Copt}(P')$. Thus, 
\[
\SqPrc_X(P') = \SqPrc_X(P' \cap \Pbad) + \SqPrc_X(P'
\setminus \Pbad ) \leq 16 \SqPrc_{\Copt}(P') +
4\SqPrc_{\Copt}(P') = 20 \SqPrc_{\Copt}(P').
\]

If $\alpha =0$ then for any point $p \in P'$, we have $(\Dist(p,X))^2
\leq n(L/4n)^2 \leq L^2/(4n)$. and thus $\SqPrc_X(P') \leq L^2/4 \leq
\SqPrc_{\Copt}(V) \leq \SqPrc_{\Copt}(P')$, since $V \subseteq P'$.

In the above analysis we assumed that the nearest neighbor data
structure returns the exact nearest neighbor. If we were to use an
approximate nearest neighbor instead, the constants would slightly
deteriorate.

\begin{lemma}
    \lemlab{good:subset}%
    Given a set $P$ of $n$ points in $\Re^d$, and parameter $k$, one
    can compute sets $P'$ and $X \subseteq P$ such that, with high
    probability, $\cardin{P'} \geq n/2$, $\cardin{X} = O(k \log^2 n)$,
    and $\SqPrc_{\Copt}(P') \geq \SqPrc_X(P')/32$, where $\Copt$ is
    the optimal set of $k$-means centers for $P$. The running time of
    the algorithm is $O(n)$ if $k = O(n^{1/4})$, and $O(n \log{(k
       \log{n})} )$ otherwise.
\end{lemma}

Now, finding a constant factor $k$-median clustering is easy. Apply
\lemref{good:subset} to $P$, remove the subset found, and repeat on
the remaining points. Clearly, this would require $O(\log{n})$
iterations. We can extend this algorithm to the weighted case, by
sampling $O(k \log^2 W)$ points at every stage, where $W$ is the total
weight of the points. Note however, that the number of points no
longer shrink by a factor of two at every step, as such the running
time of the algorithm is slightly worse.

\begin{theorem}[Clustering with more centers]
    \thmlab{k:cluster:const:rough}%
    Given a set $P$ of $n$ points in $\Re^d$, and parameter $k$, one
    can compute a set $X$, of size $O(k\log^3 n)$, such that
    $\SqPrc_X(P) \leq 32 \SqOpt(P,k)$. The running time of the
    algorithm is $O(n)$ if $k = O(n^{1/4})$, and $O(n \log{(k \log
       n)})$ otherwise.
    
    Furthermore, the set $X$ is a good set of centers for
    $k$-median.  Namely, we have that $\MdPrc_X(P) \leq 32
    \MdOpt(P,k)$.
    
    If the point set $P$ is weighted, with total weight $W$, then the
    size of $X$ becomes $O(k \log^3 W)$, and the running time becomes
    $O(n \log^2 W )$.
\end{theorem}

\section{$(1+\eps)$-Approximation for $k$-Median}
\seclab{eps:approx:k:median}

We now present the approximation algorithm using exactly $k$
centers. Assume that the input is a set of $n$ points. We use the set
of centers computed in \thmref{k:cluster:const:rough} to compute a
constant factor coreset using the algorithm of
\thmref{coreset:fast:k:median}.  The resulting coreset $\Coreset$, has
size $O(k \log^4 n )$.  Next we compute a $O(n)$ approximation to the
$k$-median for the coreset using the $k$-center (min-max) algorithm
\cite{g-cmmid-85}.  Let $C_0 \subseteq \Coreset$ be the resulting set
of centers.  Next we apply the local search algorithm, due to Arya
\etal{} \cite{agkmm-lshkm-04}, to $C_0$ and $\Coreset$, where the set
of candidate points is $\Coreset$. This local search algorithm, at
every stage, picks a center $c$ from the current set of centers
$C_{curr}$, and a candidate center $s \in \Coreset$, and swaps $c$ out
of the set of centers and $s$ into the set of centers. Next, if the
new set of centers
$C_{curr}' = C_{curr} \setminus \brc{c} \cup \brc{s}$ provides a
considerable improvement over the previous solution (i.e.,
$\MdPrc_{C_{curr}}(\Coreset) \leq
(1-\eps/k)\MdPrc_{C_{curr}'}(\Coreset)$ where $\eps$ here is an
arbitrary small constant), then we set $C_{curr}$ to be $C_{curr}'$.
Arya \etal{} \cite{agkmm-lshkm-04} showed that the algorithm
terminates, and it provides a constant factor approximation to
$\MdOptD(\Coreset, k)$, and as hence to $\MdOpt(P,k)$.  It is easy to
verify that it stops after $O(k \log{n})$ such swaps. Every swap, in
the worst case, requires considering $\cardin{\Coreset}k$
sets. Computing the price of clustering for every such candidate set
of centers takes $O\pth{ \cardin{\Coreset} k }$ time.  Thus, the
running time of this algorithm is
$O \pth{ \cardin{\Coreset}^2 k^{3} \log n}=O \pth{ k^5 \log^9 n
}$. Finally, we use the new set of centers with
\thmref{coreset:fast:k:median}, and get a $(k,\eps)$-coreset for
$P$. It is easy to see that the algorithm works for weighted
point-sets as well. Putting in the right bounds from
\thmref{k:cluster:const:rough} and \thmref{coreset:fast:k:median} for
weighted sets, we get the following.

\begin{lemma}[coreset]
    \lemlab{k:coreset:small:median}%
    Given a set $P$ of $n$ points in $\Re^d$, one can
    compute a $k$-median $(k,\eps)$-coreset $\Coreset$ of
    $P$, of size $O\pth{ (k/\eps^d)\log{n} }$, in time $O
    \pth{ n + k^5 \log^9 n }$.
    
    If $P$ is a weighted set, with total weight $W$, the
    running time of the algorithm is $O( n \log^2 W + k^5
    \log^9 W)$.
\end{lemma}

We would like to apply the algorithm of Kolliopoulos and Rao
\cite{kr-nltas-99} to the coreset, but unfortunately, their
algorithm only works for the discrete case, when the medians are
part of the input points. Thus, the next step is to generate from
the coreset, a small set of candidate points in which we can
assume all the medians lie, and use the (slightly modified)
algorithm of \cite{kr-nltas-99} on this set.

\begin{defn}[Centroid Set]
    Given a set $P$ of $n$ points in $\Re^d$, a set $\CenSet
    \subseteq \Re^d$ is an \emph{$(k,\eps)$-approximate
       centroid set} for $P$, if there exists a subset
    $C \subseteq \CenSet$ of size $k$, such that
    $\MdPrc_C(P) \leq (1+\eps)\MdOpt(P,k)$.
\end{defn}

\begin{lemma}
    \lemlab{k:median:cen:set}%
    Given a set $P$ of $n$ points in $\Re^d$, one can
    compute an $(k,\eps)$-centroid set $\CenSet$ of size
    $O(k^2\eps^{-2d}\log^2{n})$.  The running time of this
    algorithm is $O \pth{ n + k^5 \log^9 n +
       k^2\eps^{-2d}\log^2{n} }$.
    
    For the weighted case, the running time is
    $O \pth{ n \log^2 W + k^5 \log^9 W + k^2\eps^{-2d}\log^2{W} }$,
    and the centroid set is of size $O(k^2\eps^{-2d}\log^2{W})$.
\end{lemma}

\begin{proof}
    Compute a $(k,\eps/12)$-coreset $\Coreset$ using
    \lemref{k:coreset:small:median}. We retain the set
    $\CSetB$ of $k$ centers, for which $\MdPrc_\CSetB( P )
    =O( \MdOpt(P,k) )$, which is computed during the
    construction of $\Coreset$.  Further let $R =
    \MdPrc_\CSetB(P) / n$.
    
    Next, compute around each point of $\Coreset$, an exponential grid
    using $R$, as was done in \secref{construction}.  This results in
    a point set $\CenSet$ of size of $O(k^2\eps^{-2d}\log^2{n})$. We
    claim that $\CenSet$ is the required centroid set.  The proof
    proceeds on similar lines as the proof of
    \thmref{coreset:fast:k:median}.
    
    Indeed, let $\Copt$ be the optimal set of $k$ medians.
    We snap each point of $\Copt$ to its nearest neighbor in
    $\CenSet$, and let $X$ be the resulting set. Arguing as
    in the proof of \thmref{coreset:fast:k:median}, we have
    that $\cardin{\MdPrc_X(\Coreset)
       -\MdPrc_{\Copt}(\Coreset)}$ $\leq
    (\eps/12)\MdPrc_{\Copt}(\Coreset)$. On the other hand,
    by definition of a coreset, $\cardin{\MdPrc_{\Copt}(P) -
       \MdPrc_{\Copt}(\Coreset)}\leq (\eps/12)$
    $\MdPrc_{\Copt}(P)$ and $\cardin{\MdPrc_{X}(P) -
       \MdPrc_{X}(\Coreset)}\leq (\eps/12)$ $\MdPrc_{X}(P)$.
    As such, $\MdPrc_{\Copt}(\Coreset) \leq
    (1+\eps/12)\MdPrc_{\Copt}(P)$ and it follows
    \begin{equation*}
        \cardin{\MdPrc_X(\Coreset) -\MdPrc_{\Copt}(\Coreset)}
        \leq (\eps/12)(1+\eps/12)\MdPrc_{\Copt}(P) \leq
        (\eps/6) \MdPrc_{\Copt}(P).        
    \end{equation*}
    As such,
    \begin{align*}
      \MdPrc_X(P) %
      &\leq%
        \frac{1}{1-\eps/12}
        \MdPrc_X(\Coreset) \leq 2
        \MdPrc_X(\Coreset)
        \leq 2 \pth{ \MdPrc_{\Copt}(\Coreset) +
        \frac{\eps}{6} \MdPrc_{\Copt}(P)}%
      \\&%
      \leq%
      2 \pth{\pth{1+ \frac{\eps}{12}} \MdPrc_{\Copt}(P) + \frac{\eps}{6}
      \MdPrc_{\Copt}(P)}
      \leq%
      3\MdPrc_{\Copt}(P),
    \end{align*}
    for $\eps < 1$. We conclude that $\cardin{\MdPrc_X(P)
       - \MdPrc_X(\Coreset) } \leq (\eps/12)\MdPrc_X(P) \leq
    (\eps/3) \MdPrc_{\Copt}(P)$.  Putting things together,
    we have
    \begin{align*}
      \cardin{\MdPrc_{X}(P) - \MdPrc_{\Copt}(P) }
      &\leq%
        \cardin{\MdPrc_X(P) - \MdPrc_{X}(\Coreset) }
        +  \cardin{\MdPrc_{X}(\Coreset) -
        \MdPrc_{\Copt}(\Coreset) }
        +  \cardin{\MdPrc_{\Copt}(\Coreset) -
        \MdPrc_{\Copt}(P) } \\
      &\leq%
        \pth{\frac{\eps}{3} + \frac{\eps}{6} +
        \frac{\eps}{12}} \MdPrc_{\Copt}(P)
        \leq \eps\MdPrc_{\Copt}(P).
    \end{align*}
\end{proof}

We are now in the position to get a fast approximation
algorithm.  We generate the centroid set, and then we modify
the algorithm of Kolliopoulos and Rao so that it considers
centers only from the centroid set in its dynamic
programming stage.  For the weighted case, the depth of the
tree constructed in \cite{kr-nltas-99} is $O(\log{W})$
instead of $O(\log{n})$. Further since their algorithm works
in expectation, we run it independently
$O(\log(1/\delta)/\eps)$ times to get a guarantee of
$(1-\delta)$.

\begin{theorem}[\cite{kr-nltas-99}]
    \thmlab{kr:algo}%
    Given a weighted point set $P$ with $n$ points in $\Re^d$, with
    total weight $W$, a centroid set $\CenSet$ of size at most $n$,
    and a parameter $\delta > 0$, one can compute
    $(1+\eps)$-approximate $k$-median clustering of $P$ using only
    centers from $\CenSet$. The overall running time is
    $O \pth{ \Ckr n (\log k) (\log W) \log (1/\delta) }$, where
    $\Ckr = \CkrExp$.  The algorithm succeeds with probability
    $\geq 1-\delta$.

\end{theorem}
\remove{
\begin{proof}
    We need to modify the algorithm of \cite{kr-nltas-99} so that it
    considers only centers from the centroid set.  This is a
    straightforward modification of their dynamic programming stage.

    We execute the algorithm of \cite{kr-nltas-99} $M =
    O(\log(1/\delta)/\eps)$ times independently. Since the
    algorithm of \cite{kr-nltas-99} works in expectation, it
    follows by the Markov inequality, in each such execution
    the algorithm succeeds with probability larger $\eps/2$.
    Thus, it fails with probability $\leq (1-\eps/2)^M \leq
    \exp(-M \eps/2) \leq \delta$.

    As for the running time, the depth of the tree
    constructed by \cite{kr-nltas-99} is $O(\log W)$ in this
    case (instead of $O(\log n)$ as in the original
    settings). Thus, the bound on the running time follows.
\end{proof}
}

The final algorithm is the following: Using the algorithms
of \lemref{k:coreset:small:median} and
\lemref{k:median:cen:set} we generate a $(k,\eps)$-coreset
$\Coreset$ and an $\eps$-centroid set $\CenSet$ of $P$,
where $\cardin{\Coreset} = O(k\eps^{-d} \log{n})$ and
$\cardin{\CenSet} = O(k^2\eps^{-2d} \log^2{n})$. Next, we
apply the algorithm of \thmref{kr:algo} on $\Coreset$ and
$\CenSet$.
\begin{theorem}[$(1+\eps)$-approx $k$-median]
    \thmlab{fast:k:median}%
    Given a set $P$ of $n$ points in $\Re^d$, and parameter $k$, one
    can compute a $(1+\eps)$-approximate $k$-median clustering of $P$
    (in the continuous sense) in
    $O \pth{ n + k^5 \log^9 n + \Ckr k^2 \log^5 n }$ time, where
    $\Ckr = \CkrExp$ and $c$ is a constant.  The algorithm outputs a
    set $X$ of $k$ points, such that
    $\MdPrc_X(P) \leq (1+\eps)\MdOpt(P,k)$. If $P$ is a weighted set,
    with total weight $W$, the running time of the algorithm is
    $O ( n \log^2 W + k^5 \log^9 W + \Ckr k^2 \log^5 W )$.
\end{theorem}

We can extend our techniques to handle the discrete median
case efficiently as follows.
\begin{lemma}
    \lemlab{k:median:cen:set:discrete}%
    Given a set $P$ of $n$ points in $\Re^d$, one can compute a
    discrete $(k,\eps)$-centroid set $\CenSet \subseteq P$ of size
    $O(k^2\eps^{-2d}\log^2{n})$. The running time of this algorithm is
    $ \displaystyle O \pth{ n + k^5 \log^9 n + k^2\eps^{-2d}\log^2{n}
    }$ if $k \leq \eps^d n^{1/4}$ and $ \displaystyle O \pth{ n
       \log{n} + k^5 \log^9 n + k^2\eps^{-2d}\log^2{n} }$ otherwise.
\end{lemma}

\begin{proof}
    We compute a $\eps/4$-centroid set $\CenSet$, using
    \lemref{k:median:cen:set}, and let $m =\cardin{\CenSet} =
    O(k^2\eps^{-2d}\log^2{n})$. Observe that if $m > n$ then we set
    $\CenSet$ to be $P$. Next, snap every point in $P$ to its
    (approximate) nearest neighbor in $\CenSet$, using
    \corref{batch:n:n:2}.  This takes $O(n+ m n^{1/10} \log(n) )= O(n
    + k^2 n^{1/10} \eps^{-2d} \log^3 n) = O(n)$ time, if $k \leq
    \eps^d n^{1/4}$, and $O(n \log n)$ otherwise (then we use the
    data-structure of \cite{amnsw-oaann-98} to perform the nearest
    neighbor queries).  For every point $x\in\CenSet$, let $P(x)$ be
    the of points in $P$ mapped to $x$.  Pick from every set $P(x)$
    one representative point, and let $U \subseteq P$ be the resulting
    set. Consider the optimal discrete center set $\Copt$, and
    consider the set $X$ of representative points that corresponds to
    the points of $\Copt$.  Using the same argumentation as in
    \lemref{k:median:cen:set} it is easy to show that $\MdPrc_X(P)
    \leq (1+\eps)\MdOptD(P,k)$.
\end{proof}

Combining \lemref{k:median:cen:set:discrete} and
\thmref{fast:k:median}, we get the following.
\begin{theorem}[Discrete k-medians]
    \thmlab{approx:k:median:discrete}%
    One can compute an $(1+\eps)$-approximate discrete $k$-median of a
    set of $n$ points in time $\displaystyle O \pth{ n + k^5 \log^9 n
       + \Ckr k^2 \log^5 n }$, where $\Ckr$ is the constant from
    \thmref{kr:algo}.
\end{theorem}

\begin{proof}
    The proof follows from the above discussion. As for the running
    time bound, it follows by considering separately the case when
    $1/\eps^{2d} \leq 1/n^{1/10}$, and the case when $1/\eps^{2d} \geq
    1/n^{1/10}$, and simplifying the resulting expressions. We omit
    the easy but tedious computations.
\end{proof}


\section{A $(1+\eps)$-Approximation Algorithm for $k$-Means}
\seclab{eps:approx:k:mean}

\subsection{Constant Factor Approximation}
\seclab{exact:k:means:const}

In this section we reduce the number of centers to be exactly $k$.  We
use the set of centers computed by \thmref{k:cluster:const:rough} to
compute a constant factor coreset using the algorithm of
\thmref{k:means:weighted:coreset}.  The resulting coreset $\Coreset$,
has size $O(k \log^4 n)$. Next we compute a $O(n)$ approximation to
the $k$-means for the coreset using the $k$-center (min-max) algorithm
\cite{g-cmmid-85}.  Let $C_0 \subseteq \Coreset$ be the resulting set
of centers.  Next we apply the local search algorithm, due to Kanungo
\etal{} \cite{kmnpsw-lsaak-04}, to $C_0$ and $\Coreset$, where the set
of candidate points is $\Coreset$. This local search algorithm, at
every stage, picks a center $c$ from the current set of centers
$C_{curr}$, and a candidate center $s \in \Coreset$, and swaps $c$ out
of the set of centers and $c$ into the set of centers. Next, if the
new set of centers
$C_{curr}' = C_{curr} \setminus \brc{c} \cup \brc{s}$ provides a
considerable improvement over the previous solution (i.e.,
$\SqPrc_{C_{curr}}(\Coreset) \leq
(1-\eps/k)\SqPrc_{C_{curr}'}(\Coreset)$ where $\eps$ here is an
arbitrary small constant), then we set $C_{curr}$ to be
$C_{curr}'$. Extending the analysis of Arya \etal{}
\cite{agkmm-lshkm-04}, for the $k$-means algorithm, Kanungo \etal{}
\cite{kmnpsw-lsaak-04} showed that the algorithm terminates, and it
provides a constant factor approximation to $\SqOptD(\Coreset, k)$,
and as hence to $\SqOpt(P,k)$.  It is easy to verify that it stops
after $O(k \log{n})$ such swaps. Every swap, in the worst case,
requires considering $\cardin{\Coreset}k$ sets. Computing the price of
clustering for every such candidate set of centers takes
$O\pth{ \cardin{\Coreset} k }$ time.  Thus, the running time of this
algorithm is
$O \pth{ \cardin{\Coreset}^2 k^{3} \log n}=O \pth{ k^5 \log^9 n }$.

\begin{theorem}
    \thmlab{k:means:const:fast}%
    Given a point set $P$ in $\Re^d$ and parameter $k$, one can
    compute a set $X \subseteq P$ of size $k$, such that $\SqPrc_X(P)
    =O(\SqOpt(P,k))$.  The algorithm succeeds with high probability.
    The running time is $O(n + k^5 \log^9 n)$ time.

    If $P$ is weighted, with total weight $W$, then the
    algorithm runs in time $O ( n + k^5 \log^4 n$ $\log^5 W
    )$.
\end{theorem}

\remove{
\begin{proof}
    The unweighted case follows by the above discussion. As
    for the weighted case, computing a set $C$ of $O(k
    \log^2 n \log W)$ centers, which are constant factor
    approximation to the $k$-means clustering of $P$, can be
    done in $O(n \log{n} + k \log^3 n \log W)$ time, using
    \thmref{k:cluster:const:rough}. Computing a constant
    factor coreset from $C$, takes $O(n\log{ \cardin{C}}) =
    O(n (\log{n} + \log\log W))$ time, using
    \thmref{k:means:weighted:coreset}, and results in a
    coreset $\Coreset$ of size $O(\cardin{C} \log W) = O(k
    \log^2 n \log^2 W)$.

    Finally, using the local search algorithm of
    \cite{kmnpsw-lsaak-04} takes
    $O \pth{ \cardin{\Coreset}^2 k^{3} \log W } = O \pth{k^5 \log^4 n
       \log^5 W }$ time.
\end{proof}
}

\subsection{The $(1+\eps)$-Approximation}
\seclab{exact:k:means:eps}

Combining \thmref{k:means:const:fast} and
\thmref{k:means:weighted:coreset}, we get the following result
for coresets.
\begin{theorem}[coreset]
    \thmlab{k:coreset:small:means}%
    Given a set $P$ of $n$ points in $\Re^d$, one can
    compute a $k$-means $(k,\eps)$-coreset $\Coreset$ of
    $P$, of size $O\pth{ (k/\eps^d)\log{n} }$, in time $O
    \pth{ n + k^5 \log^9 n }$.
    
    If $P$ is weighted, with total weight $W$, then the
    coreset is of size $O\pth{ (k/\eps^d)\log{W} }$, and the
    running time is $O(n \log^2 W + k^5 \log^9 W)$.
\end{theorem}

\begin{proof}
    We first compute a set $\CSet$ which provides a constant factor
    approximation to the optimal $k$-means clustering of $P$, using
    \thmref{k:means:const:fast}. Next, we feed $\CSet$ into the
    algorithm \thmref{k:means:weighted:coreset}, and get a
    $(1+\eps)$-coreset for $P$, of size $O((k/\eps^d) \log{W})$.
\end{proof}

We now use techniques from \Matousek{} \cite{m-agc-00} to compute the
$(1+\eps)$-approximate $k$-means clustering on the coreset.
\begin{defn}[Centroid Set]
    Given a set $P$ of $n$ points in $\Re^d$, a set $T
    \subseteq \Re^d$ is an \emph{$\eps$-approximate centroid
       set} for $P$, if there exists a subset $C \subseteq
    T$ of size $k$, such that $\SqPrc_C(P) \leq
    (1+\eps)\SqOpt(P,k)$.
\end{defn}

\Matousek{} showed that there exists an $\eps$-approximate centroid
set of size $O(n\eps^{-d} \log(1/\eps))$.  Interestingly enough, his
construction is weight insensitive. In particular, using an
$(k,\eps/2)$-coreset $\Coreset$ in his construction, results in a
$\eps$-approximate centroid set of size $O\pth{ \cardin{\Coreset}
   \eps^{-d} \log(1/\eps)}$.

\begin{lemma}
    For a weighted point set $P$ in $\Re^d$, with total
    weight $W$, there exists an $\eps$-approximate centroid
    set of size $O(k\eps^{-2d}\log{W} \log{(1/\eps)} )$.
\end{lemma}

The algorithm to compute the $(1+\eps)$-approximation now follows
naturally. We first compute a coreset $\Coreset$ of $P$ of size
$O\pth{(k/\eps^d)\log{W} }$ using the algorithm of
\thmref{k:coreset:small:means}. Next, we compute in
$O\pth{\cardin{\Coreset} \log \cardin{\Coreset} + \cardin{\Coreset}
   e^{-d} \log{\frac{1}{\eps}}}$ time a $\eps$-approximate centroid
set $U$ for $\Coreset$, using the algorithm from \cite{m-agc-00}. We
have $\cardin{U} = O(k\eps^{-2d}\log{W} \log{(1/\eps)} )$. Next we
enumerate all $k$-tuples in $U$, and compute the $k$-means clustering
price of each candidate center set (using $\Coreset$). This takes
$O\pth{ \cardin{U}^k \cdot k \cardin{\Coreset}}$ time. And clearly,
the best tuple provides the required approximation.

\begin{theorem}[$k$-means clustering]
    \thmlab{k:means:approx}%
    Given a point set $P$ in $\Re^d$ with $n$ points, one can compute
    $(1+\eps)$-approximate $k$-means clustering of $P$ in time
    \[
    O \pth{n + k^5 \log^9 n + {k^{k+2}}{\eps^{-(2d+1)k}}
       {\log^{k+1}{n}} \log^k(\nfrac{1}{\eps})}.
    \]
    For a weighted set, with total weight $W$, the running
    time is
    \[
        O \pth{ n \log^2 W + k^5 \log^4 n \log^5 W +k^{k+2}
           \eps^{-(2d+1)k} \log^{k+1} W \log^{k} (1/\eps)}.
    \]
\end{theorem}





\section{Streaming}
\seclab{streaming}

A consequence of our ability to compute quickly a
$\pth{k,\eps}$-coreset for a point set, is that we can
maintain the coreset under insertions quickly.

\begin{observation}
    (i) If $C_1$ and $C_2$ are the $(k,\eps)$-coresets for
    disjoint sets $P_1$ and $P_2$ respectively, then
    $C_1\cup C_2$ is a $(k,\eps)$-coreset for $P_1\cup P_2$.

    (ii) If $C_1$ is $(k,\eps)$-coreset for $C_2$, and $C_2$
    is a $(k,\delta)$-coreset for $C_3$, then $C_1$ is a
    $(k,\eps+\delta)$-coreset for $C_3$.
\end{observation}

The above observation allows us to use Bentley and Saxe's technique
\cite{bs-dspsd-80} as follows. Let $P = \pth{p_1,p_2,\ldots,p_n}$ be
the sequence of points seen so far. We partition $P$ into sets $P_0,
P_1,P_2,\ldots,P_t$ such that each either $P_i$ empty or $|P_i| =
2^{i}M$, for $i>0$ and $M = O(k/\eps^d)$. We refer to $i$ as the rank
of $i$.

Define $\rho_j = \eps/\pth{c(j+1)^2}$ where c is a large
enough constant, and $1+\delta_j = \prod_{l=0}^j(1+\rho_l)$,
for $j=1,\ldots, \ceil{\lg n}$. We store a $\pth{k,\delta_j
}$-coreset $Q_{j}$ for each $P_j$.  It is easy to verify
that $1+\delta_{j} \leq 1+\eps/2$ for $j=1,\ldots, \ceil{\lg
   n}$ and sufficiently large $c$. Thus the union of the
$Q_i$s is a $(k,\eps/2)$-coreset for $P$.

On encountering a new point $p_{u}$, the update is done in
the following way: We add $p_u$ to $P_0$. If $P_0$ has less
than $M$ elements, then we are done. Note that for $P_0$ its
corresponding coreset $Q_0$ is just itself. Otherwise, we
set $Q_1' = P_0$, and we empty $Q_0$. If $Q_1$ is present,
we compute a $(k,\rho_2)$ coreset to $Q_1\cup Q'_1$ and call
it $Q'_2$, and remove the sets $Q_1$ and $Q'_1$. We continue
the process until we reach a stage $r$ where $Q_r$ did not
exist. We set $Q_r'$ to be $Q_r$. Namely, we repeatedly
merge sets of the same rank, reduce their size using the
coreset computation, and promote the resulting set to the
next rank. The construction ensures that $Q_r$ is a
$(k,\delta_r)$ coreset for a corresponding subset of $P$ of
size $2^r M$. It is now easy to verify, that $Q_r$ is a $(k,
\prod_{l=0}^j(1+\rho_l) - 1)$-coreset for the corresponding
points of $P$.

We further modify the construction, by computing a
$(k,\eps/6)$-coreset $R_i$ for $Q_i$, whenever we compute $Q_i$. The
time to do this is dominated by the time to compute $Q_i$. Clearly,
$\cup R_i$ is a $(k,\eps)$-coreset for $P$ at any point in time, and
$\cardin {\cup R_i} = O(k\eps^{-d} \log^2{n})$.

\paragraph{Streaming $k$-means}
In this case, the $Q_i$s are coresets for $k$-means clustering.
Since $Q_i$ has a total weight equal to $2^i M$ (if it is not
empty) and it is generated as a $(1+\rho_i)$ approximation, by
\thmref{k:coreset:small:means}, we have that $|Q_{i}| = O\pth{k
\eps^{-d} \pth {i+1}^{2d}(i+\log{M})}$. Thus the total storage
requirement is $O\pth{\pth{k\log^{2d+2}{n}}/\eps^d}$.

Specifically, a $(k,\rho_j)$ approximation of a subset $P_j$ of
rank $j$ is constructed after every $2^j M$ insertions, therefore
using \thmref{k:coreset:small:means} the amortized time spent for
an update is
\begin{align*}
  &\hspace{-1cm}\sum_{i=0}^{\ceil{\log{(n/M)}}} \frac{1}{2^i M}
    O\pth{\cardin{Q_i} \log^2 \cardin{P_i} + k^5 \log^9
    \cardin{P_i}}%
  \\&%
  =%
  \sum_{i=0}^{\ceil{\log{(n/M)}}}  \frac{1}{2^i M}
  O\pth{\pth{ \frac{k}{\eps^d} i^{2d}\pth{ i +
  \log M}^2 + k^5 \pth{ i +
  \log M}^9 }}
  = O \pth{ \log^2 (k/\eps) + k^5}.
\end{align*}
Further, we can generate an approximate $k$-means clustering
from the $(k,\eps)$-coresets, by using the algorithm of
\thmref{k:means:approx} on $\cup_i R_i$, with $W=n$. The
resulting running time is $O(k^5 \log^9 n +
{k^{k+2}}{\eps^{-(2d+1)k}} {\log^{k+1}{n}}
\log^k(\nfrac{1}{\eps}))$.

\paragraph{Streaming $k$-medians}
We use the algorithm of \lemref{k:coreset:small:median} for
the coreset construction. Further we use
\thmref{fast:k:median} to compute an
$(1+\eps)$-approximation to the $k$-median from the current
coreset. The above discussion can be summarized as follows.

\begin{theorem}
    \thmlab{k:means:streaming}%
    Given a stream $P$ of $n$ points in $\Re^d$ and $\eps > 0$, one
    can maintain a $(k,\eps)$-coresets for $k$-median and $k$-means
    efficiently and use the coresets to compute a
    $(1+\eps)$-approximate $k$-means/median for the stream seen so
    far. The relevant complexities are:
    \begin{itemize}
        \item Space to store the information:
        $O\pth{k\eps^{-d} \log^{2d+2}{n}}$.

        \item Size and time to extract coreset of the current set:
        $O(k\eps^{-d} \log^2 n)$.
        
        \item Amortized update time: $O\pth{ \log^2 (k/\eps)
           + k^5}$.

        \item Time to extract $(1+\eps)$-approximate $k$-means clustering:\\
        $O\pth{k^5 \log^9 n + {k^{k+2}}{\eps^{-(2d+1)k}}
           {\log^{k+1}{n}} \log^k(\nfrac{1}{\eps})}$.
        
        \item Time to extract $(1+\eps)$-approximate $k$-median
        clustering:\\ $O\pth{\Ckr k \log^7 n}$, where $\Ckr =
        \CkrExp$.
    \end{itemize}
\end{theorem}
Interestingly, once an optimization problem has a coreset, the
coreset can be maintained under both insertions and deletions,
using linear space. The following result follows in a plug and
play fashion from \cite[Theorem 5.1]{ahv-aemp-04}, and we omit the
details.
\begin{theorem}
    Given a point set $P$ in $\Re^d$, one can maintain a
    $(k,\eps)$-coreset of $P$ for $k$-median/means, using
    linear space, and in time $O(k \eps^{-d} \log^{d+2} n
    \log \frac{k \log{n}}{\eps} + k^5 \log^{10} n )$ per
    insertion/deletions.
\end{theorem}


\section{Conclusions}
\seclab{conclusions}

In this paper, we showed the existence of small coresets for
the $k$-means and $k$-median clustering.  At this point,
there are numerous problems for further research. In
particular:
\begin{enumerate}
    \item Can the running time of approximate $k$-means clustering be
    improved to be similar to the $k$-median bounds? Can one do FPTAS
    for $k$-median and $k$-means (in both $k$ and $1/\eps$)?
    Currently, we can only compute the $(k,\eps)$-coreset in fully
    polynomial time, but not extracting the approximation itself from
    it.

    \item Can the $\log{n}$ in the bound on the size of the coreset be
    removed?

    \item Does a coreset exist for the problem of $k$-median and
    $k$-means in high dimensions? There are some partial relevant
    results \cite{bhi-accs-02}.

    \item Can one do efficiently $(1+\eps)$-approximate streaming for
    the discrete $k$-median case?
    
    \item Recently, Piotr Indyk \cite{i-eabca-04} showed how to
    maintain a $(1+\eps)$-approximation to $k$-median under insertion
    and deletions (the number of centers he is using is roughly $O(k
    \log^2{\Delta})$ where $\Delta$ is the spread of the point set). 
    It would be interesting to see if one can extend our techniques to
    maintain coresets also under deletions. It is clear that there is
    a linear lower bound on the amount of space needed, if one assume
    nothing. As such, it would be interesting to figure out what are
    the minimal assumptions for which one can maintain
    $(k,\eps)$-coreset under insertions and deletions.
\end{enumerate}

\section*{Acknowledgments}

The authors would like to thank Piotr Indyk, Satish Rao and
Kasturi Varadarajan for useful discussions of problems
studied in this paper and related problems.


\newcommand{\etalchar}[1]{$^{#1}$}
 \providecommand{\CNFX}[1]{ {\em{\textrm{(#1)}}}}
  \providecommand{\tildegen}{{\protect\raisebox{-0.1cm}{\symbol{'176}\hspace{-0.03cm}}}}
  \providecommand{\SarielWWWPapersAddr}{http://sarielhp.org/p/}
  \providecommand{\SarielWWWPapers}{http://sarielhp.org/p/}
  \providecommand{\urlSarielPaper}[1]{\href{\SarielWWWPapersAddr/#1}{\SarielWWWPapers{}/#1}}
  \providecommand{\Badoiu}{B\u{a}doiu}
  \providecommand{\Barany}{B{\'a}r{\'a}ny}
  \providecommand{\Bronimman}{Br{\"o}nnimann}  \providecommand{\Erdos}{Erd{\H
  o}s}  \providecommand{\Gartner}{G{\"a}rtner}
  \providecommand{\Matousek}{Matou{\v s}ek}
  \providecommand{\Merigot}{M{\'{}e}rigot}
  \providecommand{\Hastad}{H\r{a}stad\xspace}
  \providecommand{\CNFCCCG}{\CNFX{CCCG}}
  \providecommand{\CNFBROADNETS}{\CNFX{BROADNETS}}
  \providecommand{\CNFESA}{\CNFX{ESA}}
  \providecommand{\CNFFSTTCS}{\CNFX{FSTTCS}}
  \providecommand{\CNFIJCAI}{\CNFX{IJCAI}}
  \providecommand{\CNFINFOCOM}{\CNFX{INFOCOM}}
  \providecommand{\CNFIPCO}{\CNFX{IPCO}}
  \providecommand{\CNFISAAC}{\CNFX{ISAAC}}
  \providecommand{\CNFLICS}{\CNFX{LICS}}
  \providecommand{\CNFPODS}{\CNFX{PODS}}
  \providecommand{\CNFSWAT}{\CNFX{SWAT}}
  \providecommand{\CNFWADS}{\CNFX{WADS}}


\appendix



\section{Fuzzy Nearest-Neighbor Search in Constant Time}
\apndlab{fast:n:n}

Let $X$ be a set of $m$ points in $\Re^d$, such that we want to answer
$\eps$-approximate nearest neighbor queries on $X$. However, if the
distance of the query point $q$ to its nearest neighbor in $X$ is
smaller than $\delta$, then it is legal to return any point of $X$ in
distance smaller than $\delta$ from $q$. Similarly, if a point is in
distance larger than $\Delta$ from any point of $X$, we can return any
point of $X$. Namely, we want to do nearest neighbor search on $X$,
when we care only for an accurate answer if the distance is in the
range $[\delta,\Delta]$.

\begin{defn}
    Given a point set $X$ and parameters $\delta, \Delta$ and $\eps$,
    a data structure $D$ answers \emph{$(\delta,\Delta,\eps)$-fuzzy
       nearest neighbor} queries, if for an arbitrary query $q$, it
    returns a point $x \in X$ such that
    \begin{enumerate}
        \item If $\Dist(q, X) > \Delta$ then $x$ is an arbitrary point
        of $X$.

        \item If $\Dist(q, X) < \delta$ then $x$ is an arbitrary point
        of $X$ in distance smaller than $\delta$ from $q$.

        \item Otherwise, $\dist{q x} \leq (1+\eps)\Dist(q,X)$.
    \end{enumerate}
\end{defn}

In the following, let $\rho = \Delta/\delta$ and assume that
$1/\eps = O(\rho)$.  First, we construct a grid $G_{\Delta}$
of size length $\Delta$, using hashing and the floor
function, we throw the points of $X$ into their relevant
cells in $G_{\Delta}$. We construct a NN data structure for
every non-empty cell in $G_\Delta$. Given a query point $q$,
we will compute its cell $c$ in the grid $G_{\Delta}$, and
perform NN queries in the data-structure associated with
$c$, and the data-structures associated with all its
neighboring cells, returning the best candidate generated.
This would imply $O(3^d)$ queries into the cell-level NN
data-structure.

Consider $Y$ to be the points of $X$ stored in a cell $c$ of
$G_{\Delta}$. We first filter $Y$ so that there are no
points in $Y$ that are too close to each other. Namely, let
$G$ be the grid of side length $\delta \eps/(10d)$. Again,
map the points of $Y$ into this grid $G$, in linear time.
Next, scan over the nonempty cells of $G$, pick a
representative point of $Y$ from such a cell, and add it to
the output point set $Z$. However, we do not add a
representative point $x$ to $Z$, if there is a neighboring
cell to $c_x$, which already has a representative point in
$Z$, where $c_x$ is the cell in $G$ containing $x$. Clearly,
the resulting set $Z \subseteq Y$ is well spaced, in the
sense that there is no pair of points of $Z$ that are in
distance smaller than $\delta\eps/(10d)$ from each other. As
such, the result of a $(\delta,\Delta, \eps)$-fuzzy NN query
on $Z$ is a valid answer for a equivalent fuzzy NN query
done on $Y$, as can be easily verified. This filtering
process can be implemented in linear time.

The point set $Z$ has a bounded stretch; namely, the ratio
between the diameter of $Z$ and the distance of the closet
pair is bounded by $\Delta/(\delta\eps/(10d)) = O(\rho^2)$.
As such, we can use a data structure on $Z$ for nearest
neighbors on point set with bounded stretch \cite[Section
4.1]{h-rvdnl-01}. This results in a quadtree $T$ of depth
$O( \log( \rho ) ) \leq c \log{\rho}$, where $c$ is
constant.  Answering NN queries, is now done by doing a
point-location query in $T$, and finding the leaf of $T$
that contains the query point $q$, as every leaf $v$ in $T$
store a point of $Z$ which is a $(1+\eps)$-approximate
nearest neighbor for all the points in $c_v$, where $c_v$ is
the region associated with $v$. The construction time of $T$
is $O(\cardin{Z} \eps^{-d} \log \rho )$, and this also bound
the size of $T$.

Doing the point-location query in $T$ in the naive way, takes $O(
\depth(T)) = O(\log{\rho})$ time. However, there is a standard
technique to speed up the nearest neighbor query in this case to $O(
\log \depth(T))$ \cite{aeis-eards-99}. Indeed, observe that one can
compute for every node in $T$ a unique label, and furthermore given a
query point $q=(x,y)$ (we use a 2d example to simplify the exposition)
and a depth $i$, we can compute in constant time the label of the node
of the quadtree $T$ of depth $i$ that the point-location query for $q$
would go through.  To see that, consider the quadtree as being
constructed on the unit square $[0,1]^2$, and observe that if we take
the first $i$ bits in the binary representation of $x$ and $y$,
denoted by $x_i$ and $y_i$ respectively, then the tuple $(x_i, y_i,
i)$ uniquely define the required node, and the tuple can be computed
in constant time using bit manipulation operators.

As such, we hash all the nodes in $T$ with their unique tuple id into
a hash table. Given a query point $q$, we can now perform a binary
search along the path of $q$ in $T$, to find the node where this path
``falls of'' $T$. This takes $O(\log \depth(T))$ time.

One can do even better. Indeed, we remind the reader that
the depth of $T$ is $c \log{\rho}$, where $c$ is a constant.
Let $\alpha = \ceil{(\log \rho )/(20d r)} \leq
(\log{\rho})/(10d r)$, where $r$ is an arbitrary integer
parameter. If a leaf $v$ in $T$ is of depth $u$, we continue
to split and refine it till all the resulting leaves of $v$
lie in level $\alpha \! \ceil{u/\alpha}$ in $T$. This would
blow up the size of the quadtree by a factor of $O(
(2^d)^\alpha ) = O(\rho^{1/r})$.  Furthermore, by the end of
this process, the resulting quadtree has leaves only on
levels with depth which is an integer multiple of $\alpha$.
In particular, there are only $O(r)$ levels in the resulting
quadtree $T'$ which contain leaves.

As such, one can apply the same hashing technique described above to
$T'$, but only for the levels that contains leaves.  Now, since we do
a binary search over $O(r)$ possibilities, and every probe into the
hash table takes constant time, it follows that a NN query takes
$O(\log r)$ time.

We summarize the result in the following theorem.
\begin{theorem}
    \thmlab{fast:n:n}%
    Given a point set $X$ with $m$ points, and parameters $\delta,
    \Delta$ and $\eps >0$, then one can preprocess $X$ in $O( m
    \rho^{1/r} \eps^{-d} \log( \rho/\eps) )$ time, such that one can
    answer $(\delta,\Delta,\eps)$-fuzzy nearest neighbor queries on
    $X$ in $O(\log r)$ time. Here $\rho = \Delta/\delta$ and $r$ is an
    arbitrary integer number fixed in advance.
\end{theorem}

\begin{theorem}
    \thmlab{batch:n:n}%
    Given a point set $X$ of size $m$, and a point set $P$ of size $n$
    both in $\Re^d$, one can compute in $O(n+ m n^{1/4} \eps^{-d}
    \log(n/\eps) )$ time, for every point $p \in P$, a point $x_p \in
    X$, such that $\dist{p x_p} \leq (1+\eps) \Dist(p, X) + \tau/n^3$,
    where $\tau = \max_{p \in P} \Dist(p,X)$.
\end{theorem}

\begin{proof}
    The idea is to quickly estimate $\tau$, and then use
    \thmref{fast:n:n}. To estimate $\tau$, we use a similar
    algorithm to the closet-pair algorithm of Golin \etal{}
    \cite{grss-sracp-95}. Indeed, randomly permute the points
    of $P$, let $p_1,\ldots, p_n$ be the points in permuted
    order, and let $l_i$ be the current estimate of $r_i$,
    where $r_i = \max_{j=1}^{i} \Dist(p_i, X)$ is the
    maximum distance between $p_1,\ldots, p_i$ and $X$.  Let
    $G_i$ be a grid of side length $l_i$, where all the
    cells contains points of $X$, or their neighbors are
    marked.  For $p_{i+1}$ we check if it contained inside
    one of the marked cells. If so, we do not update the
    current estimate, and set $l_{i+1} = l_i$ and $G_{i+1} =
    G_i$. Otherwise, we scan the points of $X$, and we set
    $l_{i+1} = 2\sqrt{d} \Dist(p_{i+1},X)$, and we recompute
    the grid $G_{i+1}$. It is easy to verify that $r_{i+1}
    \leq l_{i+1}$ in such a case, and $r_{i+1} \leq 2
    \sqrt{d} l_{i+1}$ if we do not rebuild the grid.

    Thus, by the end of this process, we get $l_n$, for
    which $l_n/(2\sqrt{d}) \leq \tau \leq 2\sqrt{d} l_n$, as
    required. As for the expected running time, note that if
    we rebuild the grid and compute $\Dist(p_{i+1},X)$
    explicitly, this takes $O(k)$ time.  Clearly, if we
    rebuild the grid at stage $i$, and the next time at
    stage $j>i$, it must be that $r_i \leq l_i < r_j\leq
    l_j$.  However, in expectation, the number of different
    values in the series $r_1, r_2, \ldots, r_n$ is
    $\sum_{i=1}^{n} 1/i =O(\log{n})$. Thus, the expected
    running time of this algorithm is $O(n + k \log{n})$, as
    checking whether a point is in a marked cell, takes
    $O(1)$ time by using hashing.

    We know that $l_n/(2\sqrt{d}) \leq \tau \leq 2\sqrt{d} l_n$.  Set
    $\delta = l_n/(4d^2 n^5)$, $\Delta= 2 \sqrt{d} l_n$ and build the
    $(\delta,\Delta, \eps)$-fuzzy nearest neighbor data-structure of
    \thmref{fast:n:n} for $X$. We can now answer the nearest neighbor
    queries for the points of $P$ in $O(1)$ per query.
\end{proof}

\begin{corollary}
    Given a point set $X$ of size $m$, a point set $P$ of size $n$
    both in $\Re^d$, and a parameter $D$, one can compute in
    $O(n+ m n^{1/10} \eps^{-d} \log(n/\eps) )$ time, for every point
    $p \in P$ a point $x_p \in P$, such that:
    \begin{itemize}
        \item If $\Dist(p,X) > D$ then $x_p$ is an arbitrary
        point in $X$.

        \item If $\Dist(p,X) \leq D$ then $\dist{p x_p} \leq
        (1+\eps) \Dist(p, X) + D/n^4$.
    \end{itemize}

    \corlab{batch:n:n:2}
\end{corollary}


\end{document}